\documentclass{LMCS}

\def\dOi{11(1:8)2015}
\lmcsheading%
{\dOi}
{1--27}
{}
{}
{Mar.~29, 2013}
{Mar.~16, 2015}
{}

\ACMCCS{[{\bf Theory of computation}]: Logic---Constructive Mathematics}

\input{amssym.def}
\input{amssym.tex}
\usepackage{hyperref}

\theoremstyle{plain}\newtheorem{lemma}[thm]{Lemma}

\newcommand{\R}{{\Bbb R}}
\newcommand{\N}{{\Bbb N}}

\newcommand{\Q}{{\Bbb Q}}
\newcommand{\Z}{{\Bbb Z}}

\newcommand{\Ct}{\mathsf{Ct}}

\newcommand{\rinf}{\rightarrow \infty}

\begin{document}
\title[The extensional realizability model of continuous functionals]{The extensional realizability model of continuous functionals and three weakly non-constructive classical theorems }
\author[D.Normann]{Dag Normann}
\address{Department of Mathematics, The University 
of Oslo, P.O. Box 1053, Blindern N-0316 Oslo, Norway} \email{dnormann@math.uio.no} 

\begin{abstract}We investigate wether three statements in analysis, that can be proved classically, are realizable in the realizability model of extensional continuous functionals induced by Kleene's second model $K_2$. We prove that a formulation of the Riemann Permutation Theorem as well as the statement that all partially Cauchy sequences are Cauchy cannot be realized in this model, while the statement that the product of two anti-Specker spaces is anti-Specker can be realized. \end{abstract}  
\maketitle
\section{Introduction}\label{introduction}
\subsection{Discussion}
The background motivation for the results obtained in this paper is the desire to understand the relative strength of classical theorems in mathematics in a constructive context. Some statements will, when added to a constructive theory, transform the theory to a classical one. We consider a statement to be weakly non-constructive when this is not the case. Weakly non-constructive statements may be identified as such when they are satisfied by models of constructive mathematics not satisfying classical logic, but are not constructively provable themselves.

We will consider three weakly non-constructive statements and the interpretation of them in the extensional realizability model induced by Kleene's second model $K_2$, see below or the main text for a discussion of $K_2$. For two of the statements we will prove that they fail in this model. The consequence will be that they cannot be proved from other weakly non-constructive principles true in the same model modulo e.g. $HA^\omega$ or extensional Martin-L\"of type theory with one universe. The third statement is classically equivalent to the finite dimensional Tychonov theorem for metric spaces. We will prove a technical theorem showing that, in a certain sense, this statement is true in the model. We return to a precise formulation of this statement below, and discuss how we interpret it in the model.

Ishihara introduced, implicitly in  \cite{Ishihara1} and explicitly in  \cite{Ishihara2}, the principle known as BD-$\N$: Let $\{a_n\}_{n \in \N}$ be a sequence of natural numbers, and assume that for all functions $f:\N \rightarrow \N$ there is an $n$, depending on $f$, such that
$\forall k \geq n (a_{f(k)} < k)$. Then $\{a_n\}_{n \in \N}$ is bounded. It is not hard to prove BD-$\N$ classically. BD-$\N$ is an example of a weakly non-constructive statement, but it turned out to be rather strong as such. Berger and Bridges \cite{BD1,BD2} proved that the Riemann Permutation Theorem is a consequence of BD-$\N$ and Bridges \cite{Bridges} proved that BD-$\N$ implies that the product of two anti-Specker spaces is anti-Specker (see the formal definition of anti-Specker spaces in Subsection \ref{aS}). Fred Richman (unpublished) introduced the concept of partially Cauchy sequences and proved that, as a consequence of BD-$\N$, all partially Cauchy sequences are Cauchy.

Lubarsky and Diener \cite{LD} showed that neither of these statements imply BD-$\N$ and that the closure under products of anti-Specker spaces is not outright provable constructively.  We refer to \cite{LD} for a further discussion of these results.

Kleene \cite{Kleene.59II} defined a typed structure of \emph{countable functionals}, and to this end he introduced a way to let functions $f:\N \rightarrow \N$ code partial continuous functionals $F:\N^\N \rightarrow \N$. This has been modified to a partial operator $f \bullet g$ of type $(\N^\N)^2 \rightarrow \N^\N$, an operator that organizes $\N^\N$ to a \emph{partial combinatorial algebra}. This algebra is known as Kleene's second model $K_2$. Partial combinatorial algebras will in turn generate realizability models for constructive mathematics, and it is the extensional realizability model based on $K_2$ and \emph{namings} or \emph{modest assemblies} that we will work with here.

The paper is self contained in the sense that we will introduce the mentioned realizability model to the extent needed to make our technical results precise. Thus one does not need to be familiar with constructive mathematics as such in order to read the paper. \emph{Constructive mathematics} is not a precise concept, anyhow. Our results will shed some light on formal non-classical logics accepting a realizability model based on $K_2$ and compatible with our interpretations of the statements we consider.

Even though we are investigating aspects of constructive mathematics, we will use full set theory with classical logic in our proofs. In particular, we will use proofs by contradiction in order to prove two of our results.

\subsection{Main Results}
We will prove: 
\begin{enumerate}
\item The extensional realizability model induced by $K_2$ does realize that the product of two anti-Specker spaces is anti-Specker.
\item The extensional realizability model induced by $K_2$ does not realize that all partially Cauchy sequences are Cauchy.
\item The Riemann Permutation Theorem is not realized by the extensional realizability model induced by $K_2$.
\end{enumerate}
\subsection{Outline}
The paper is organized as follows:
\begin{itemize}
\item[-] In Section \ref{background} we introduce Kleene's second model to the extent we need it, and we turn the three main results into precise mathematical statements.

\item[-]In Section \ref{Sec3} we prove that $K_2$ will realize that a metric space $X$ is anti-Specker if and only if $X$ is compact, we introduce a kind of realizer of compactness, a \emph{compactness base}, and we use this to construct a realizer of anti-Speckerhood of a product from realizers of anti-Speckerhood of the two factors. It is in this form that we prove result 1.

\item[-]In Section \ref{Sec4} we prove results 2 and 3, and in Section \ref{Sec5} we summarize what we have obtained.
\end{itemize}

\subsection{Acknowledgements}
In early 2011, Thomas Streicher made me aware of the problem concerning the realizability of the Riemann Permutation Theorem in the extensional $K_2$-model, and I worked unsuccessfully on this problem for some time. Then, in May/June 2012 Robert Lubarsky asked me the same question for partially Cauchy and for the closure under products of the class of anti-Specker spaces. This was when we were both visiting fellows at the \emph{Isaac Newton Institute for Mathematical Sciences} in Cambridge, UK, participating in the program \emph{Semantics and Syntax: A Legacy of Alan Turing}. I am grateful to the Newton Institute for inviting me and letting me take part in this program.

During the preparation of this paper, I benefitted from further discussions with Robert Lubarsky and Thomas Streicher, and their comments on two informal working notes were most helpful. Thomas Streicher and Mart\'in Escard\'o  gave helpful comments on the choice of notation and the exposition in  a first draft of the paper version.

Two anonymous referees gave valuable feedback, for which I am grateful. Their comments were very useful in the preparation of the final version.
\section{Background}\label{background}
\subsection{Kleene's Second model}
Kleene's second model $K_2$ is an organization of the \emph{Baire space} $\N^\N$ into a partial applicative structure, actually a partial combinatorial algebra, and has its origin in Kleene \cite{Kleene.59II}.

We let $\langle \cdot , \ldots , \cdot \rangle$ be a standard sequence numbering. 
\begin{itemize}
\item[-]If $f:\N \rightarrow \N$ and $n \in \N$, we let $\bar f(n) = \langle f(0) , \ldots , f(n-1)\rangle$, where $\bar f(0)$ is the sequence number of the empty sequence.
\item[-]
If $g:\N \rightarrow \N$ and $n \in \N$, we let $\langle n,g\rangle \in \N^\N$ be defined by $$\langle n,g\rangle(k) = \left \{ \begin{array}{ccc} n& {\rm if} & k=0\\ g(k-1)& {\rm if} &k > 0\end{array} \right .$$
\item[-] Given $f$ and $g$ in $\N^\N$ we let $f * g = f(\bar g(n))-1$ for the least $n$ such that $f(\bar g(n)) > 0$ if there is one such $n \in \N$.
We let $f*g$ be undefined if there is no such $n$.
\item[-] We let $f \bullet g = \lambda k.f*\langle k, g\rangle$ when this function is total.
We let $f \bullet g$ be undefined otherwise.
\end{itemize}
\begin{defi}\label{definition2.1}\hfill
\begin{enumerate}[label=\alph*)]
\item A \emph{naming} of a set $X$ will be a pair $(A,\nu)$ where $A \subseteq \N^\N$ and $\nu:A \rightarrow X$ is onto. We say that $f$ is a \emph{name} of $\nu(f)$.
\item If $(A,\nu)$ is a naming of $X$, the \emph{induced topology} on $X$ is the finest topology making $\nu$ continuous.
\item If $(A_X,\nu_X)$ is a naming of $X$, $(A_Y,\nu_Y)$ is a naming of $Y$, $ \phi:X \rightarrow Y$ and $f \in \N^\N$, we let $f$ be a name of $\phi$ if  whenever $g \in A_X$ we have that $f \bullet g \in A_Y$ and
$\nu_Y(f \bullet g) = \phi(\nu_X(g))$. We then also say that $f$ is \emph{tracking} $\phi$.
\item Given the notation from c), we let $Z = X \rightarrow Y$ be the set of functions with names and $(A_Z,\nu_Z)$ the corresponding naming.
\end{enumerate}
\end{defi}

\noindent If $X$ is a structured set, the structure must be reflected in the naming. In this paper we will be concerned with namings of metric spaces, and we will see later how to deal with this.

There are of course many namings of the same set or structure. In order to express that constructions are uniform in the naming, we need the concept of a \emph{reduction} of one naming to another:
\begin{defi}Let $(A,\nu)$ and $(B,\eta)$ be two namings of the set $X$, and let $h:\N \rightarrow \N$.

We say that $h$ is a \emph{reduction} of $(A,\nu)$ to $(B,\eta)$ if $h \bullet g \in B$ for all $f \in A$, and then $$\nu(f) = \eta(h \bullet f).$$

\end{defi}
\begin{rem}\hfill
\begin{enumerate}[label=\alph*)]
\item[a)]The topology on $X$ defined in Definition \ref{definition2.1} b) is homeomorphic to the quotient topology on the set of equivalence classes in $A$ under the equality relation for $\nu$.
\item[b)] \emph{Naming} is a term borrowed from the TTE-approach to computational analysis initiated by Weihrauch, see \cite{Weihrauch} for a book exposition.
\item[c)] Another term for a naming is \emph{a modest assembly}, where we write
$$f \Vdash x$$
instead of $\nu(f) = x$. Then $f$ is often called a realizer of $x$ or a representation of $f$. We will not adopt this terminology, using the word ``realizer" in a slightly different way.
\item[d)]
Topological spaces that are quotients of namings are also known as {\bf QCB}-spaces, see Battenfield, Schr\"oder and Simpson \cite{BSS} for the definition and  Schr\"oder \cite{Schroder} for the characterization of {\bf QCB} via namings. 
\item[e)] When $Z$ is a set with a naming $\rho$, $Z$ is actually some space of functions, i.~e.~a subset of a set  $X \rightarrow Y$, and $\rho$ is defined from namings $\nu$ and $\eta$ for $X$ and $Y$ resp., then the elements of the domain of $\rho$ are often called \emph{associates}, in line with Kleene's \cite{Kleene.59II} original use of this term. We will stick to using ``names'' here.
\end{enumerate}
\end{rem}
\begin{exa} Let $\nu_\N(f) = f(0)$. This gives a naming of $\N$, and seemingly we may use $A_\N = \N^\N$. However, using many names for each integer may generate unwanted problems, so we let $f \in A_\N$ if $f(k+1) = 0$ for all $k$.

Then every finite type $t$ over the base type $\N$ will have a canonical interpretation in $K_2$, and indeed, the interpretation of $t$ will be the Kleene-Kreisel functionals $\Ct(t)$ of type $t$.
\end{exa}

Following Kleene's original definition, we may also use $\N$ for naming purposes, letting $\langle \N,id\rangle$ be a naming of $\N$. For the sake of notational simplicity we will adopt this, and use the operator $(f,g) \mapsto f*g$ when $f$ is the name of a function with arguments having function names $g$, but values in $\N$.

$K_2$ will be used to form a \emph{realizability model}  for second order arithmetic, and we will be interested in the extensional version. We will not give the full details of the construction of this model, but we will be precise when it comes to the three statements we are interested in.

The philosophy behind a realizability model is that the truth of a statement $\Phi$ may be realized by some object $\phi$. An implication $\Phi \Rightarrow \Psi$ will be realized by a function mapping a realizer for $\Phi$ to a realizer for $\Psi$. In our context, the sets of realizers of $\Phi$ and $\Psi$ will be subject to suitable namings, and then a realizer for the implication will be a function with a name in the sense of Kleene's second model.

We will not be precise about which formal language we use to express our statements, but it will be sufficiently typed. There will be types representing $\N$ and $\N^\N$, and this will suffice to express two of our statements. The third statement is a general statement about metric spaces, and we will simply explain how this is to be understood in the present context. Thus, the reader is not assumed to be familiar with constructions of realizability models in general.

A closed formula with parameters from $\N$ or $\N^\N$ will be realized by a set of functions or functionals, and the formula is \emph{true in the model} if the set of realizers is nonempty. 
\begin{itemize}
\item[-]A true quantifier free formula will be realized by the list of parameters, while a false quantifier free formula will have no realizers.
\item[-] A purely universal formula that is true will also be realized by the parameters appearing in the formula. 
\item[-] An existential formula will be realized by a pair consisting of a witness to the truth of the formula and a realizer of this truth.
\item[-] A universal formula will be realized by a function in the model mapping any interpretation of the variable in question to a realizer of the corresponding instance. It is here, and in the definition of realizers of implications,  we must be precise in  what we mean by a function in the model.
\item[-] Disjunctions are realized by pairs consisting of one of the disjuncts and a realizer of the chosen one, and conjunctions will be realized by pairs of realizers in the obvious way.
\end{itemize}
There are two natural ways we may use $K_2$ to form a realizability model for second order arithmetic, the \emph{intensional} one and the \emph{extensional} one. In the intensional model we let the realizers be functions in $\N^\N$, and when we need realizers that map functions to realizers or realizers to realizers, we use the application operator in $K_2$ to interpret a function as a partial functional. We will not be concerned with the intensional model here, except in one example.

In the extensional model, we will use true functions as  realizers of a formula and $K_2$ indirectly via a naming of the realizers as described above. Let us give an example of what we mean:
\begin{exa}Consider the statement
$$\forall f \exists n (f(n) \leq f(n+1) \vee (n > 0 \wedge f(n-1) > f(n).$$ 
Classically, a realizer will be a functional $F$ such that for all functions $f$, either $f(F(f)) \geq f(F(f)+1)$ or $F(f) > 0$ and $f(F(f) - 1) > f(F(f))))$.

There are many  such functionals $F$, and the continuous ones are definable within $K_2$. We will use them as extensional realizers. A function $g$  will be a name of  the realizer $F$  of the statement if $g * f = F(f)$ for all $f$. 

(If we were to be truly faithful to the idea of realizers, $F(f)$ should not only give information about an $n$ that satisfies the disjunction, but also of which disjunct that is satisfied by $n$.)
\end{exa}
Ishihara \cite{Ishihara1,Ishihara2} introduced the following principle known as BD-$\N$ , strong enough to prove all the three statements we will consider in this paper:
\begin{exa} We let $f$, $g$ and $h$ range over $\N^\N$ and let $n$, $m$ etc. range over $\N$. BD-$\N$ is  the following implication
$$\forall g(\forall f \exists n \forall k \geq n (g(f(k)) < k) \Rightarrow \exists n \forall k (g(k) < n)).$$
An intensional realizer of the assumption
$$\forall f \exists n \forall k \geq n (g(f(k)) < k)$$ will be a function $h$ such that $$\forall f\forall k \geq (h * f)(g(f(k)) < k).$$
A realizer for the conclusion will simply be a number $n$ such that $\forall k (g(k) < n)$, i.~e.~ an upper bound of the function $g$.

It is easy to find an upper bound like this from $g$, from the identity function $f = {\rm id}$ and from an intensional realizer $h$ of the assumption, as the following argument shows: 

There will be numbers $t$ and $n$ such that $$h(\bar f(t)) = n+1.$$
The significance is that for all $f'$ extending $\bar f(t)$ we have 
$$\forall k \geq n (g(f'(k)) < k).$$
Let $n_0 = \max\{t,n\}$.
For  $m \in \N$, let $f_m$ be an extension of $\bar f(t)$ such that $f_m(n_0) = m$.
We then have that $$g(m) = g(f_m(n_0)) < n_0.$$
Thus $n_0$ is an upper bound for $g$.
This upper bound depends on $h$, and not only on the functional defined from $h$, and indeed, this has to be the case.

 Lietz and Streicher \cite{LS} showed that we cannot do this in a continuous way, depending only of the functional encoded by $h$. We will give a direct proof of this as an example of how one may argue about this model.

 Assume that we for some continuous $M$ defined on the set of all pairs $(F,g)$, where $F$ is continuous  of type 2 and
$$(\ast)\:\: \forall f \forall k \geq F(f)(g(f(k)) < k),$$
  have that  $$\forall k (g(k) < M(F,g)).$$
 What this actually means is that there is a $K_2$-name $\alpha \in \N^\N$ of $M$ such that whenever $h$ is a $K_2$-name of $F$ satisfying $(\ast)$, then $$((\alpha \bullet h)* g) $$ is defined and independent of the choice of $h$ (among the names of $F$), and $$\forall k(g(k) < ((\alpha \bullet h) * g)).$$
 Now choose $g$ to be the constant zero function and $F$ to be the constant 1 functional. Then $(\ast)$ is satisfied.

 Let $k = M(F,g)$ and let $h$ be defined on the set of finite sequences $\tau$ by
 $$h(\tau) = \left\{ \begin{array}{ccc}0 & {\rm if} & lh(\tau) \leq k+1 \\ 2 & {\rm if} & lh(\tau) > k+1\end{array} \right.$$ where $lh(\tau)$ is the length of the sequence. $h$ is a name for $F$. Then 
$$((\alpha \bullet h)* g) = k+1.$$
 We will only need finite segments $\bar h(a)$ of $h$ and $\bar g(b)$ of $g$ to secure this value. We will obtain a contradiction by constructing a $g_1$ extending $\bar g(b)$ and an $F_1$ with a $K_2$-name $h_1$ extending $\bar h(a)$ such that $(\ast)$ holds for $F_1$ and $g_1$, but where $k$ is not an upper bound for $g_1$.

 Without loss of generality, we may assume that $a = b \geq k+1$.
 Let $g_1(n) = 0$ if $n < a$, while $g_1(n) = a$ if $n \geq a$.
 Let $F_1(f) = 1$ if $\bar f(k+2) < a$, and $F_1(f) = a+1$ otherwise.
 Then $F_1$ has a $K_2$ name $h_1$ defined by
 $$h_1(\tau) = \left \{ \begin{array}{ccccc}0& {\rm if} &lh(\tau) \leq k+1& & \\2 & {\rm if} & lh(\tau) > k+1 & {\rm and}& \bar \tau(k+1) < a\\a+2 & {\rm if} & lh(\tau) > k+1 & {\rm and} & \bar \tau(k+1) \geq a\;.\end{array} \right .  $$
 The pair $(F_1,g_1)$ will satisfy $(\ast)$, and $h_1$ will be an extension of $\bar h(a)$. 
 $M(F_1,g_1) \geq a+1$ since $a+1$ is the least proper upper bound of $g_1$. This contradicts that $M(F_1,g_1) = k$ since $g_1$ extends $\bar g(a)$ and $F_1$ has a name extending $\bar h(a)$.
 The assumption leading to this contradiction was the existence of $M$ with the given properties. Thus there is no such extensional functional $M$.
 This ends our example.
\end{exa}

\subsection{The Riemann Permutation Theorem}
The {Riemann Permutation Theorem} is the following classical result:
\begin{prop}[RPT]
Let $\{x_i\}_{i \in \N}$ be a sequence of reals such that for every permutation $p$ of $\N$ we have that the series
$$\sum_{i = 0}^\infty x_{p(i)}$$ converges.
Then the series $$\sum_{i = 0}^\infty |x_i|$$ converges, i.e. the series converges absolutely.
\end{prop}
We will prove that {\bf RPT} is not realizable in the extensional model of continuous functionals. In this section we will make our claim precise, and then we will prove it in Section \ref{RPT-proof}.

First, we will simplify the statement by restricting it to sequences of rational numbers in $\Q$. We will assume the existence of an underlying enumeration of $\Q$, but in order to save notation, we will treat $\Q$ as a discrete set of the same computational status as $\N$. Arithmetical equalities and inequalities on $\Q$ are of course decidable.

Then, when we discuss {\bf RPT} (and later, partially Cauchy), a \emph{sequence} is  a function $\bar x:\N \rightarrow \Q$.

Classically, a sequence has a limit if and only if it is Cauchy, and a series has a limit if and only if the sequence of partial sums is Cauchy. As is customary in constructive mathematics, we will use this as the \emph{definition} of \emph{ having a limit}.

We will use the notation $\bar x = \{x_i\}_{i \in \N}$ as a convention
without always stating this equality explicitly.

The sequence $\bar x$ is Cauchy if
$$\forall n \exists m \forall i \geq m \forall j \geq m (|x_i - x_j| < 2^{-n}).$$
A realizer for this will be a function $f:\N \rightarrow \N$ such that
$$\forall n \forall i \geq f(n) \forall j \geq f(n)(|x_i - x_j| < 2^{-n})$$
together with the sequence itself.
A more standard terminology, that we will adopt, is that $f$ is a modulus (of Cauchyness) for $\bar x$.

Next we will see what  a realizer for the assumption in {\bf RPT} will be. Ignoring that $\bar x$ is a parameter, a classical realizer will be a function $F$ that maps a permutation $p$ to a modulus for $$\left \{\sum_{i < n}x_{p(i)} \right \}_{n \in \N}.$$
Actually, $F$ should take a pair $(p,g)$ as argument, where $p$ is a permutation and $g$ is a realizer of this fact. However, a realizer for $p$ being a permutation will be the inverse, and since the set of permutations $p$ and the set of pairs $(p,p^{-1})$ are homeomorphic, considering  realizers in this case will only cause extra notational complexity.

Thus given $\bar x$, a realizer for the assumption in {\bf RPT} for $\bar x$ will be a continuous functional mapping a permutation $p$ to a modulus $F(p)$ for the corresponding series, where $F$ is continuous if it has a $K_2$-name $f$ satisfying that $f \bullet p$ is a modulus for $$\left \{\sum_{i < n}x_{p(i)}\right \}_{n \in \N}$$ whenever $p$ is a permutation.

A realizer for {\bf RPT} will then be a continuous functional $M$ that to any sequence $\bar x$ and a realizer $F$ for the assumption of {\bf RPT} on $\bar x$ produces a modulus $M(\bar x , F)$ for the Cauchyness of $$\left \{\sum_{i < n}|x_i|\right \}_{n \in \N}.$$ The functional $M$ is continuous according to our model if there is a function $h$ such that whenever $\bar x$ is a sequence and whenever $f$ is a name for a realizer $F$ for the assumption in {\bf RPT} for $\bar x$, then $(h \bullet \bar x )\bullet f = M(\bar x,F)$, and then we let $h$ be a $K_2$-name for $M$. Here, we of course consider $\bar x$ as a sequence of numbers via the given enumeration of $\Q$.

We will prove  that there is no such continuous realizer for {\bf RPT}.
\subsection{Partially Cauchy implies Cauchy}
Another classical theorem (in the sense of the logic needed to prove it) is that every partially Cauchy sequence is indeed Cauchy. While the Riemann Permutation Theorem is an established theorem in classical analysis, the theorem that all partially Cauchy sequences are Cauchy  must be viewed as a constructed example of a statement with only a slight non-constructive content. The definition of partially Cauchy is due to Fred Richman (unpublished), and is as follows
\begin{defi} Let $\bar x$ be a sequence of rational numbers.
$\bar x$ is \emph{partially Cauchy} if we for every total function $f \geq id$ on $\N$ have that
$$\lim_{n \rinf}diam\{x_n , \ldots , x_{f(n)}\} = 0.$$
\end{defi}

A realizer for $\bar x$ being partially Cauchy will be a continuous functional $F$ with values in $\N^\N$ and  defined on the set of $f \geq id$ such that $F(f)$ is an increasing function for each $f \geq id$, and 
$$\forall f \geq id\; \forall n\; \forall k \geq F(f)(n)(diam\{x_k , \ldots , x_{f(k)}\} < 2^{-n}).$$
We use the $K_2$-naming of the set of realizers in the same way as before.
Then a realizer of {\bf pC} $ \Rightarrow$ {\bf C} will be a continuous functional $M$ that to any rational sequence $\bar x$ and any continuous  realizer $F$ of $\bar x$ being partially Cauchy provides a modulus $M(\bar x , F)$ for $\bar x$ being Cauchy.

We do not give the detailed explanation of what the continuous realizers will be in this case, assuming that the reader can fill in the details when needed.

In an informal note circulated to a few people, a sequence $\bar x$ was erroneously defined  as partially Cauchy when $$diam\{x_n , \ldots , x_{f(n)}\}$$ is Cauchy whenever $f \geq id$. This error was observed by Thomas Streicher.

This is actually equivalent to the sequence being partially Cauchy, and Streicher pointed out in a private communication that this is evenly constructively so. This result has no impact on the rest of this paper, and we leave the proof as an exercise for the interested reader.

In Subsection \ref{RPT-proof} we will prove that our model does not realize that every partially Cauchy sequence is Cauchy. 
\subsection{Anti-Specker spaces}\label{aS} An \emph{anti-Specker space} will be a metric space $X$ satisfying a statement that in a roundabout way expresses that $X$ is sequentially compact. In order to realize properties of anti-Specker spaces in $K_2$ we have to be precise in what an anti-Specker space is, and how to model one  in the sense of $K_2$. The concept of an anti-Specker space is not uniquely defined in the literature. We take our definition from \cite{LD}, restricting it to metric spaces where the distance function is bounded by 1.
\begin{defi}
Let $\langle X , d\rangle$ be a metric space, $x \in X$ and $\{x_n\}_{n \in \N}$ a sequence from $X$.
\newline
We say that $\{x_n\}_{n \in \N}$ \emph{avoids} $x$ if $$\exists \epsilon > 0 \exists n \in \N\forall m \geq n(d(x,x_m) > \epsilon).$$
\end{defi}

\begin{defi}
Let $\langle X,d\rangle$ be a metric space where $d$ takes values in $[0,1]$.
Extend $X$ to $X^\ast = X \cup\{\ast\}$  with one extra point $\ast$, and the metric $d$ to $d^\ast$ by letting $d^\ast(x,\ast) = 1$ for all $x \in X$.
We say that $X$ is an \emph{anti-Specker space} if whenever $\{x_n\}_{n \in \N}$ is a sequence from $X^\ast$ avoiding all elements $x \in X$ there is a number $n$ such that $x_m = \ast$ for all $m \geq n$.
\end{defi}

We will use \emph{namings} to model metric spaces in general, and anti-Specker spaces in particular, in $K_2$. 

First of all, we need a decent naming of the reals. Ignoring the need of trivial coding, we will use the \emph{ signed digit representation}:
\begin{defi} Let $\hat \R = \Z \times \{-1,0,1\}^\N$ considered as functions defined on $\N$.
For $f \in\hat  \R$, let $$\rho(f) = f(0) + \sum_{n = 1}^\infty f(n)\cdot2^{-n}$$
\end{defi}

The point is that any other continuous naming of $\R$ can be factorized through $\hat \R$ via $\rho$.
\begin{defi} A \emph{metric naming} will be an ordered tuple $\langle A,\nu,X,d,\hat d\rangle$ where
\begin{enumerate}
\item $\langle X,d\rangle$ is a metric space in the ordinary sense.
\item $A \subseteq \N^\N$ and $\nu:A \rightarrow X$ is a naming.
\item $\hat d:A^2 \rightarrow \hat \R$ is continuous and $\rho (\hat d(f,g)) = d(\nu(f),\nu(g))$ for all $f,g \in A$.
\end{enumerate}
\end{defi}

We do not have an exact reference for the following two observations, but there is nothing original in the arguments, see e.g. \cite{Schroder} for a more systematic treatment:
\begin{prop}\label{proposition1}Let $\langle A , \nu , X , d , \hat d\rangle$ be a metric naming. The quotient topology on  $X$ induced by $\nu$ is finer than the topology induced by the metric $d$.\end{prop}
\begin{proof}
Let $x,y \in X$ with $d(x,y) = a$. Let $f,g \in A$ with $\nu(x) = f$ and $\nu(g) = y$. Then $$\rho(\hat d(f,g)) = a.$$
Let $\epsilon > 0$. Then $U = \{g \in \hat R \mid |\rho(g) - a| < \epsilon\}$ is open in $\hat R$, and is a union of equivalence classes with respect to the equality relation of $\rho$.
Then $\{(f_1,f_2) \in A^2 \mid \hat d(f_1,f_2) \in U\}$ is open in $A^2$ and is the union of equivalence classes with respect to the equality relation of $\nu^2$.

The $\nu^2$-range of this set will then be open in the quotient topology, and coincides with the $d$-inverse of $\langle  a - \epsilon, a + \epsilon\rangle$. This shows that $d$ is continuous with respect to the quotient topology, and the proposition is established.

\end{proof}
\begin{prop}\label{proposition2} Let $\langle X,d\rangle$ be a metric space. If $\langle X,d\rangle$ allows a metric naming, then there is one metric naming $\langle B,\eta,X,d,\bar d\rangle$ such that every other metric naming of $\langle X,d\rangle$ can be reduced to it.
Moreover, the quotient topology induced by $\eta$ will coincide with the metric topology on $X$.
\end{prop}
\begin{proof}If a topological space $X$ has a naming, it will be \emph{herditarily Lindel\"of}, meaning that every open covering of a subset of $X$ has a countable subcovering. This is a consequence of the fact that the domain of the naming has a countable base. As a consequence of Proposition \ref{proposition1} and the assumption we then see that if $\langle X,d\rangle$ allows a metric naming, there will be a map $s \mapsto O_s$ from the set of finite sequences from $\N$ to the set of open subsets of $X$ such that
\begin{itemize}
\item[-]$O_\varepsilon = X$, where $\varepsilon$ is the empty sequence.
\item[-] If the sequence $s$ has length $n>0$, then the diameter of $O_s$ is $\leq 2^{(1-n)}$

\item[-] For each sequence $s$,
$$\{O_{sn} \mid n \in \N\}$$ is an open covering of $O_s$.
\end{itemize}
Let $$B = \{f \in \N \rightarrow \N \mid \bigcap_{n \in \N}O_{\bar f(n)} \neq \emptyset\}$$ and let $\eta(f)$ be the unique element in this intersection when $f \in A$.
By construction, the map $(f,g) \mapsto d(\eta(f),\eta(g))$, mapping $B^2 \rightarrow \R$, is continuous, and then there will be a continuous $\bar d:B^2 \rightarrow \hat R$ such that
$$\rho(\bar d(f,g)) = d(\eta(f),\eta(g))$$ for all $(f,g) \in B^2$. (Every continuous function from a zero-dimensional space to $\R$ can be factorized through $\hat R$ via $\rho$.)
Let us first show that if $U \subseteq X$ is open in the quotient topology, then $U$ is open in the metric topology. So, let $U$ be open, let $x \in U$ and let $f \in B$ be such that $x = \eta(f)$.
Then there is a number $n$ such that 
$$\forall g \in B(\forall i < n(g(i) = f(i)) \Rightarrow \eta(g) \in U).$$
Let $s = \bar f(n)$. By the construction of $B$ and $\eta$ we see from the above that $O_s \subseteq U$. It follows that $U$ is open in the metric topology.

Now, let $A$ with $\nu$ and $\hat d$ be any other metric naming. We will construct a continuous function $H:A \rightarrow B$ whose name will be a reduction of $\langle A , \nu\rangle$ to $\langle B,\eta\rangle$. Actually, we will let $f \in A$ and we will construct $g = H(f)$ by recursion, where we use the word \emph{construct} in a rather liberal way.

For $m \in \N$, let $A_{\bar f(m)} = \{g \in A \mid \bar f(m) = \bar g(m)\}$. We will find an increasing sequence $m_0,m_1, \ldots , m_k,\ldots$ and decide the value of $g(k)$ on the basis of $\bar f_{m_k}$.
We let $m_0 = 0$.

Assume that $m_k$ and $\bar g(k)$ is determined such that $\nu(f) \in O_{\bar g(k)}$.
Then, for some $m_{k+1}$ and $n$ we have that $$A_{\bar f_{m_{k+1}}} \subseteq \nu^{-1}(O_{\bar g(k)n}).$$
Select one such pair, and let $g(k) = n$.

If we at each stage choose the least possible $m_{k+1}$ and then the least possible $n$, we may view the construction of $H$ as a map from finite sequences to finite sequences, so it has a name $h$ in $K_2$. That $h$ is a reduction of $\langle A,\nu\rangle$ to $\langle B,\eta\rangle$ is then trivially verified.
\end{proof}

\begin{defi} A naming satisfying the conclusions of Proposition \ref{proposition2} will be called a \emph{universal metric naming of $\langle X,d\rangle$}.
\end{defi}
From now on we will restrict ourselves to the situation where the metric $d$ takes values in $[0,1]$. We then let $\langle X^\ast, d^\ast \rangle$ be as above. Without loss of generality we will assume that we have chosen a fixed naming $\langle A,\nu,\hat d\rangle$   in such a way that we can extend it to a set $A^\ast = A \cup\{f_\ast\}$ and $\nu^\ast:A \cup\{f_\ast\} \rightarrow X \cup \{\ast\}$ such that we can continuously and uniformly in the choice of naming  decide, for $f \in A^\ast$, if $f \in A$ or $f = f_\ast$. We may, for instance, let $A \subset (\N \rightarrow \N_{> 0})$ and $f_\ast$ be the constant zero function.

In order to simplify the notation, we will use $\nu$ for both namings, and we will not distinguish, in notation, between the metrics on $X$ and on $X \cup \{\ast\}$.
We will not assume that the naming is universal in general.
\begin{defi}
Using the notation from above we define:
Let $f \in A$ and let $\{f_i\}_{i \in \N}$ be a sequence from $A \cup \{f_\ast\}$.
Then
$\{f_i\}_{i \in \N}$ \emph{avoids} $f$,  if $\{\nu(f_i)\}_{i \in \N}$ avoids $\nu(f)$.\end{defi}

This actually means that there are numbers $n$ and $m$ such that
$d(\nu(f),\nu(f_i)) \geq 2^{-n}$ for all $i \geq m$. 
We then have, independent of the choice of naming:
\begin{itemize} 
\item[-]$X$ is an anti-Specker space if for all sequences $\{f_i\}_{i
  \in \N}$ from $A \cup\{f_\ast\}$ that avoids all $f \in A$ we have
  that $f_i = f_\ast$ for all but finitely many $i$.
\end{itemize}
Using the standard definitions of realizability, we see:

\begin{obs}
Let $X$ and $X^\ast$ be as above, $\nu$ a metric naming of $X$.
\begin{enumerate}[label=(\alph*]
\item[a)] Let $\{f_i\}_{i \in \N}$ be a sequence from $A \cup
  \{\ast\}$. $K_2$ realizes that $\{f_i\}_{i \in \N}$ avoids all $f
  \in A$ if there is a continuous functional $H:A \rightarrow \N$ such
  that if $f \in A$ and $H(f) = \langle n,m \rangle$ then
  $d(\nu(f),\nu(f_i)) \geq 2^{-n}$ for all $i \geq m$.\medskip
 
  \noindent\emph{$H$ will then be a realizer of this fact, and we use the
  $K_2$-naming of partial continuous functionals to define the
  topology on the set of realizers.}\medskip

\item[b)] $K_2$ realizes, with respect to the given naming, that $X$ is an anti-Specker space if there is a continuous functional $M$ defined on the set of pairs $(\{f_i\}_{i \in \N},H)$  of sequences $\{f_i\}_{i \in \N}$ from $A \cup \{f_\ast\}$ avoiding all $f \in A$ and the realizers $H$ of this fact  such that $$\forall j \geq M(\{f_i\}_{i \in \N},H)(f_j = f_\ast)$$
 
  \noindent\emph{$M$ will be the realizer, and the topology on the set
    of realizers is defined from the $K_2$-naming.}
\end{enumerate}
\end{obs}
\begin{rem} Notice that we do not insist that the realizer $H$ in a) respects the equivalence induced by $\nu$. One good reason is that in the case when $\langle X,d\rangle$ is a connected metric space, only the constant functions into $\N$ would then be possible.
\end{rem}

\begin{rem} We have described what it means for $K_2$ to realize that a given naming names an anti-Specker space. In the sequel we will prove that this is independent of the choice of the naming. We will then use this as our definition of $K_2$ realizing that $X$ is an anti-Specker space. 

Strictly spoken, the theory of metric spaces is higher order, but
since we will avoid a general discussion of what it means for $K_2$ to
realize arithmetical statements of order beyond 2, we stop our
analysis of how $K_2$ relates to anti-Speckerhood here.
\end{rem}
The problem under discussion is whether the property that the product of two anti-Specker spaces is anti-Specker is realized by $K_2$. It is easy to define a naming of the product $X \times Y$ of two named metric spaces $X$ and $Y$,  using the metric on the product space where the distance between two pairs is the maximum of the distances in each coordinate. The precise problem that we tacle will be:
\begin{prob}\label{problem1}  Assume that $M_X$ and $M_Y$ are realizers of $X$ and $Y$ being anti-Specker spaces with respect to some given namings.
Will we then have a realizer for $X \times Y$ being an anti-Specker
space with respect to our chosen naming of the product, and if this is
the case, can we find one such realizer continuously from $M_X$ and
$M_Y$?
\end{prob}
We will give a positive answer to this problem.
In Section \ref{Sec3} we will show that we can choose the naming of $X \times Y$ in such a way that this problem has a positive solution.

A realizer for the statement that the product of two anti-Specker spaces is an anti-Specker space is at type level 4, so a bonus will be that we will have constructed a new example of a functional of type 4 of some mathematical interest.
\section{Anti-Specker spaces and compactness}\label{Sec3}
In this section we will see that $K_2$ will realize that $X$ is an anti-Specker space if and only if $X$ is compact, and we will use an elaboration of this to prove that our model realizes that the product of two anti-Specker spaces is anti-Specker. We will stick to the notation explained in Section \ref{Sec3}.
\begin{lemma}\label{lemma3.1} Let $\{x_i\}_{i \in \N}$ be a sequence from $X \cup \{\ast\}$ that avoids all points in $X$. Let $\{f_i\}_{i \in \N}$ be a sequence where $f_i$ is a name for $x_i$ for each $i \in \N$.
Then $K_2$ will realize that $\{f_i\}_{i \in \N}$ avoids all $f \in A$.
\end{lemma}
\begin{proof}
Let $x \in X$, $f \in A$ and assume that $x = \nu(f)$. Then there are numbers $n$ and $m$ such that $$\forall i \geq m (d(x,x_i) \geq 2^{-n})$$
Then there is an open neighborhood $B$ of $x$ with radius $2^{-(n+1)}$ such that for all $y \in B$ and all $i \geq m$ we have that $d(y,x_i) \geq 2^{-(n+1)}$.
As a consequence we see that there is an open covering $\{O_{m,n}\}_{m,n \in \N}$ of $X$ such that for each $x \in O_{m,n}$ and $i \geq m$ we have that $d(x,x_i) \geq 2^{-n}$.

Let $U_{n,m}$ be the $\nu$-inverse of $O_{n,m}$. This family will be an open covering of $A$.  $A$ has a basis $\{C_i\}_{i \in \N}$ of sets that are both closed and open (clopen), and without loss of generality we may assume that each $C_i$ is a subset of some designated $U_{n_i,m_i}$. Let $f \in A$. We let $f \in V_{n_i,m_i}$ if $i$ is minimal such that $f \in C_i$. Then $V_{n,m} \subseteq U_{n,m}$, $V_{n,m}$ is clopen, the $V_{n,m}$'s are pairwise disjoint and $\{V_{n,m} \mid n,m \in \N\}$ covers $A$.

We let $H(f) = \langle m,n\rangle$ on $V_{m,n}$. $H$ will  be continuous, and a $K_2$-realizer for the fact that $\{f_i\}_{i \in \N}$ avoids all $f \in A$. \end{proof}
%
One consequence of Lemma \ref{lemma3.1} is that if $X$ has a naming for which $K_2$ realizes that $X$ is anti-Specker, then $X$ is sequentially compact, and thus compact.
We will  prove the strong version of the converse, but need some notation first.

For the rest of this chapter, $\sigma$, $\tau$ etc. will denote finite partial functions from $\N$ to $\N$ (and not just finite sequences). Moreover, $n$, $m$, $i$, $j$, $k$ etc. will denote elements of $\N$. $f$, $g$ and $h$ will denote total functions. We use $\sqsubseteq$ as the subfunction-ordering.

When $\sigma$ is a finite sequence, i.e. defined on an initial segment of $\N$, we identify $\sigma$ with its sequence number in order to obtain notational simplicity.
\begin{defi} Given the metric naming $\langle A,\nu, X, d, d^\ast\rangle$, $\sigma$ and $n$, we let
$$O^X_{\sigma,n} = \{x \in X \mid \exists  f \in A ( \sigma \sqsubseteq f \wedge d(x,\nu(f)) < 2^{-n})\}$$
We will omit the superscript $X$ when there can be no confusion.
\end{defi}
Notice that $O_{\sigma,n}$ will be an open set. We cannot normally tell, from the available data, if $O_{\sigma,n}$ is empty or not.
\begin{lemma} \label{lemma3.2}
Let $X$ be compact. Then $K_2$ realizes that $X$ is anti-Specker (with respect to any naming).
\end{lemma}
\begin{proof}
Let $\{f_i\}_{i \in \N}$ be a sequence from $A \cup \{f_\ast\}$ that avoids all $f \in A$ and let $H$ be a realizer of this fact. Let
$$M(\{f_i\}_{i \in \N},H) = \mu m \forall i \geq m (f_i = f_\ast)$$
The value of $M$ is actually independent of $H$ and we can compute $M(\{f_i\}_{i \in \N},H)$ from $\{f_i\}_{i \in \N}$ and any upper bound of $M(\{f_i\}_{i \in \N},H)$.
Thus it suffices to show that for any name $h$ for a realizer $H$ of the assumption that  the sequence avoids all $f \in A$, we can determine an upper bound for the value of $M$ from finite information from $h$. So, let $h$ be given.

For each $f \in A$ there is a finite sequence $\sigma \sqsubseteq f$ such that $h(\sigma) > 1$. If $h(\sigma) = \langle n,m \rangle + 1$ we of course have that $\nu(f) \in O_{\sigma,n}$, but we will also have that for $i \geq m$, $\nu(f_i) \not \in O_{\sigma,n}$ (due to the choice of $H$ and $h$).

We get an open covering of $X$ this way, and since $X$ is compact, there will be a finite subcovering
$$O_{\sigma_1,n_1}, \ldots , O_{\sigma_k,n_k}$$ where $h(\sigma_j) = \langle n_j,m_j\rangle + 1$ for $j = 1 , \ldots , k$.

If we now consider the least initial segment of $h$ for which we can find an open finite covering as above, we see that $$\max\{m_1 , \ldots , m_k\}$$ will be an upper bound for $M$.  From this we can compute the value of $M$ itself.
\end{proof}
\begin{rem}
This proof is nonuniform and noneffective in the sense that we
actually have to know when a finite set of open sets of the form
$O_{\sigma,n}$ is an open covering of $X$ or not. Thus we have not yet
proved that the following corollary can be realized in $K_2$:
\end{rem}
\begin{cor}
If $K_2$ realizes that $X$ and $Y$ are anti-Specker, then $K_2$ realizes that $X \times Y$ is anti-Specker.
\end{cor}
We will now set forth to improve this corollary to a positive solution of Problem \ref{problem1}.
\begin{defi} A \emph{ base covering} of $X$ will be an enumerated sequence $\{(\sigma_k,n_k)\}_{k \in \N}$ such that
for each $f \in A$   there is a $k \in \N$ with $\sigma_k \sqsubseteq f$.
\end{defi}
The base covering will represent the covering $$\{O_{\sigma_k,n_k} \mid k \in \N\},$$ and for $X$ to be compact, it suffices that each such covering has a finite subcovering.
\begin{defi} Let $X$ be a named metric space as above.
A \emph{compactness base} for $X$ is a set $\mathcal B$ of finite sets $\Theta = \{(\sigma_1,n_1), \ldots , (\sigma_k,n_k)\}$ such that 
\begin{itemize}
\item for each $\Theta \in {\mathcal B}$ the set
$$\hat \Theta = \{O_{\sigma_1,n_1} , \ldots , O_{\sigma_k,n_k}\}$$
is a covering of $X$ 
\item  whenever  $\{(\tau_{i},m_{i})\}_{i \in \N}$
 is a  base covering of $X$, then there is a $$\Theta = \{(\sigma_1,n_1) , \ldots , (\sigma_k,n_k)\} \in {\mathcal B}$$ such that for each $j \leq k$ there is an $i \in  \N$ with $\tau_{i} \sqsubseteq \sigma_j$ and $m_{i} \leq n_j$.
\end{itemize}
We say that $\Theta$ \emph{subcovers} $\{(\tau_i,n_i)\}_{i \in \N}$ when $\Theta$ is as above.
\end{defi}
\begin{rem} When $\Theta$ subcovers $\{(\tau_i,n_i)\}_{i \in \N}$, it
  is in essence a witness to the fact that $\{O_{\tau_i,n_i}\}_{i \in
    \N}$ has a finite subcovering.
\end{rem}

\begin{lemma}\label{lemma3.7}
We can, continuously in an enumeration of a compactness base ${\mathcal B}$ for $X$ construct a name $\alpha$ for the object realizing that $X$ is anti-Specker.
\end{lemma}
\begin{proof}
We use the proof of Lemma \ref{lemma3.2}, but with some care. Let $h$ be a name for a realizer $H$ for the assumption that $\{\nu(f_i)\}_{i \in \N}$ avoids all elements in $X$. Then
$$\{(\sigma,n) \mid \exists m (h(\sigma) = \langle n,m \rangle + 1)\}$$ is a base covering.
We search $\mathcal B$ for a $$\Theta = \{(\sigma_1,n_1) , \ldots , (\sigma_k,n_k)\}$$ such that for each $j = 1 , \ldots , k$ there is a \emph{sequence} $\tau_j \sqsubseteq \sigma_j$ and numbers $n'_j \leq n_j$ and $m_j$ such that $$h(\tau_j) = \langle n'_j,m_j\rangle + 1.$$
Then $$M(\{f_i\}_{i \in \N},H) \leq \max\{m_j \mid 1 \leq j \leq k\}.$$
There is an  $l$ such that all information from $h$ used in this search is in $\bar h(l)$. We let $$\alpha(\bar h(l)) = M(\{f_i\}_{i \in \N},H) + 1$$ when $l$ is sufficiently large, $\alpha(\bar h(l)) = 0$ otherwise.

For any $h$ and $l$ we can decide if $\bar h(l)$ approximates a name for a realizer $H'$ of some sequence $\{f'_i\}_{i \in \N}$ avoiding all $f \in A$ well enough to determine the value of $M(\{f'_i\}_{i \in \N},H')$ in the way we have described, so our $\alpha$ will be a total function. 
\end{proof}
\begin{lemma} \label{lemma3.8}
Let $X$ be compact and let $\alpha$ be a name for the functional $M$ realizing that $X$ is anti-Specker.
Then, continuously in $\alpha$, we can construct an enumeration of a compactness base for $X$.
\end{lemma}
\begin{proof}
Let $x_i = \ast$ for all $i$. Then all continuous functionals $H:A \rightarrow \N$ will be  realizers for the fact that this sequence avoids all $f \in A$. We have that $M(\{f_\ast\}_{i \in \N},H) = 0$ for all $H$ that are total on $A$.

Any finite sequence can be extended to a name for a realizer of the fact that $\{f_\ast\}_{i \in \N}$ avoids all $f \in A$.

From $\alpha$ we can construct a name $\beta$ for
$$\lambda H. M(\{f_\ast\}_{i \in \N},H).$$
Let $\beta(\tau) = 1$ with $\beta(\tau_1) = 0$ for all proper subsequences $\tau_1$ of $\tau$. Let $t \in \N$ be so large that $\beta(\tau)$ is definable from $\alpha$ and $\{f_\ast\}_{i < t}$. Let $\sigma_1 , \ldots , \sigma_k$ be the $\sqsubseteq$-minimal elements in the set of all sequences $\sigma$ for which $\tau(\sigma) > 0$.  
For $j \leq k$ we have that $\tau(\sigma_k) = \langle n_k,m_k \rangle + 1$ for some $m_k$ and $n_k$.\medskip

\noindent{\em Claim.} 
$\{O_{\sigma_1,n_1} , \ldots , O_{\sigma_k,n_k}\}$ is  a covering of $X$.
\medskip

\noindent\emph{Proof of claim}
If this is not the case, we may construct an alternative sequence $\{y_i\}_{i \in \N}$ satisfying
\begin{itemize}
\item[-] $y_i = \ast$ if $i \leq \max\{t,m_j \mid 1 \leq j \leq k\}$
\item[-] $y_i \not \in \bigcup\{O_{\sigma_j,n_j} \mid 1 \leq j \leq k\}$ if $i = \max\{t,m_j \mid 1 \leq j \leq k\}+1$
\item[-]$y_i = \ast$ if $i > \max\{t,m_j \mid 1 \leq j \leq k\}+1$.
\end{itemize}
$\tau$ will also approximate a name for a realizer $H_1$ of the fact that this new sequence, with a naming $\{g_i\}_{i \in \N}$, avoids all $f \in A$.  $M( \{g_i\}_{i \in \N}, H_1) > 0$ by the construction. This contradicts that $\{g_i\}_{i \in \N}$ extends $\{f_\ast\}_{i < t}$, that $\tau$ approximates a name for a realizer for the statement that $\{g_i\}_{i \in \N}$ avoids all $f \in A$ and that $\beta(\tau) = 1$.
This ends the proof of the claim.\medskip

We will show that the set of all $$\Theta_\tau = \{(\sigma_1,n_1), \ldots , (\sigma_k,n_k)\}$$ constructed this way is a compactness base  for $X$. The covering part is established in the claim.

Let $\{(\tau_i,m_i)\}_{i \in \N}$ be a base covering of $X$.
Define $h$ by $h(\sigma) = \langle 0,m_i\rangle + 1$ for the least $i$ such that $\tau_i \sqsubseteq \sigma$ if there is such $i$, and we let $h(\sigma) = 0$ otherwise.
Then $h$ is a name for a realizer of the statement that $\{f_\ast\}_{i \in \N}$ avoids all $f \in A$, so for some minimal  $n$, $\beta(\bar h(n)) > 0$. If we let $\tau = \bar h(n)$, $\Theta_\tau$ will subcover $\{(\tau_i,m_i)\}_{i \in \N}$.
\end{proof}
Before we can construct a realizer of the fact that the product of two anti-Specker spaces is an anti-Specker space, we need to construct of a naming of $X \times Y$ from namings of $X$ and $Y$. Along the way, we repeat the definition of $X \times Y$ as a metric space.
\begin{defi} Let $\langle A,\nu,X,d_X,d_X^\ast \rangle$ and $\langle B, \eta,Y,d_Y,d_Y^\ast \rangle$ be two metric namings.
\begin{enumerate}[label=(\alph*]
\item[a)]
Let $$C = A \times B = \{\langle f,g\rangle \mid f \in A \wedge g \in B\}$$
where \begin{itemize}
\item[-] $\langle f,g\rangle(2n) = f(n)$
\item[-] $\langle f,g\rangle(2n+1) = g(n)$
\end{itemize}
\item[b)] Let $d((x,y),(x',y')) = \max\{d_X(x,x'),d_Y(y,y')\}$ and let $$d^\ast(\langle f,g \rangle , \langle f',g'\rangle) = \max{} ^\ast(\{d_X^\ast(f,f'),d_Y^\ast(g,g')\})$$ where $\max ^\ast$ is a pre-chosen lifting of $\max$ from $\R$ to $\R^\ast$.
\item[c)] For $\langle f,g\rangle \in C$ we let $\theta(\langle f,g\rangle) = (\nu(f),\eta(g))$.
\end{enumerate}
We then let $\langle C,\theta,X \times Y,d,d^\ast \rangle$ be our metric naming of $X \times Y$ as a metric space.
\end{defi}

We have chosen the pairing of functions in such a way that it extends to pairings of partial functions as well. In particular, the pair of two finite sequences will be a finite partial function, even if the two sequences do not have the same length.

We will now let $A$, $X$, $\nu$, $B$, $Y$, $\eta$ etc. be as above. We let $Z = X \times Y$ and we let $C$, $\theta$ etc. also be as above.
Observe that $$O^Z_{\langle \sigma , \tau \rangle , n} = O^X_{\sigma , n} \times O^Y_{\tau,n}.$$
\begin{defi} Let $\sigma,\tau,n,m$ be given.
We let $(\sigma,n) \otimes (\tau,m) = (\langle \sigma,\tau\rangle, \min\{n,m\})$
\end{defi}

\begin{defi}\label{def3.11} Let ${\mathcal B}^X$ be a compactness base for $X$ and ${\mathcal B}^Y$ be a compactness base for $Y$.
We define ${\mathcal B}^Z$ as the set of sets
$$\{(\sigma_i,n_i) \otimes (\tau_{i,j},m_{i,j}) \mid 1 \leq i \leq k \wedge 1 \leq j \leq l_i\}$$ where $$\{(\sigma_1 , n_i) , \cdots , (\sigma_k , n_k)\} \in {\mathcal B}^X$$ and $$\{(\tau_{i,1},m_{i,1}) , \ldots , (\tau_{i,l_i},m_{i,l_i})\} \in {\mathcal B}^Y$$ for each $i = 1 , \ldots , k$.\end{defi}
\begin{lemma}\label{lemma3.12} Let ${\mathcal B}^Z$ be constructed from the compactness bases ${\mathcal B}^X$ and ${\mathcal B}^Y$ as in Definition \ref{def3.11}.
Then ${\mathcal B}^Z$ is a compactness base for $Z$.
\end{lemma}
\begin{proof}
It is easy to see that $\hat \Theta$ is a covering of $Z$ whenever $\Theta \in {\mathcal B}^Z$. The other property requires some more work.

Let $\{(\delta_i,n_i)\}_{i \in \N}$ be a base covering of $Z$, where $\delta_i = \langle \sigma_i,\tau_i\rangle$. (All partial functions on $\N$ can be viewed as  pairs of partial functions.) 
For each $f \in A$, the set
$$\{(\tau_i,n_i) \mid \sigma_i \sqsubseteq f\}$$
is a base covering of $Y$. For each $f$, let $\Theta^Y_f \in {\mathcal B}^Y$ subcover $\{(\tau_i,n_i) \mid \sigma_i \sqsubseteq f\}$.

For each $f \in A$, let $$\sigma_f = \bigsqcup\{\sigma_i \mid \sigma_i \sqsubseteq f \wedge \exists (\tau,m) \in \Theta_f (\tau_i \sqsubseteq \tau \wedge n_i \leq m)\}$$ and let $n_f$ be the maximal value of the corresponding numbers $n_i$.
Then $$\{(\sigma_f,n_f)\mid f \in A\}$$ is a base covering of $A$.

Let $\Theta^X$ be a subcover of $\{(\sigma_f,n_f)\mid f \in A\}$. For each $(\sigma,m) \in \Theta$, pick one $f \in A$ such that $\sigma_f \sqsubseteq \sigma$ and $m \geq n_f$, and then let $(\tau,n) \in \Theta^Y_f$. 

It remains to show that $\Theta$ constructed from $\Theta^X$ and the finite choice of $\Theta^Y_f$'s as above subcovers  the given base cover. So, let $(\sigma,m) \otimes (\tau,n)$ be in $\Theta$ as constructed, via $f \in A$. Then $\sigma_f \sqsubseteq \sigma$ and $m \geq n_f$.

Since $\Theta_f$ subcovers $\{(\tau_i,n_i)\mid \sigma_i \sqsubseteq f\}$, there is one $i$ with $\sigma_i \sqsubseteq f$, $\tau_i \sqsubseteq \tau$ and $n_i \leq n$. By construction, $\sigma_i \sqsubseteq \sigma_f$ and $n_i \leq n_f \leq m$.
It follows that $\delta_i \sqsubseteq \langle \sigma,\tau\rangle$ and that $n_i \leq \min\{m,n\}$, and this is exactly what is required.
\end{proof}

Combining the lemmas in this section we now have a proof of
\begin{thm}\label{3.15} Let $X$ and $Y$ be two compact meric spaces with namings, and let $M_X$ and $M_Y$ be realizers of the facts that these two spaces are anti-Specker. Then, continuously in $M_X$ and $M_Y$ we can find a realizer of the fact that $X \times Y$ (with the chosen metric and naming) is anti-Specker.
\end{thm}
\begin{rem} We interpret this as
\begin{center}{\em 
The extensional realizability model based on $K_2$ realizes that the product of two anti-Specker spaces is anti-Specker,}
\end{center}
although we have not expressed the statement that the product of two anti-Specker spaces is an anti-Specker space in such a way that it is precise to say what it means to realize it.

A close analysis of the proof of Theorem \ref{3.15} will reveal that the continuous function we construct will be computable, the links between the compactness base and the realizer of anti-Speckerhood are computable and the product of compactness bases is computable.\end{rem}
In this section we have based our argument on one specific way of naming a set $X \times Y$ from namings of $X$ and $Y$ and on one specific metric on the product space. These particular choices are not essential for the result. This follows from  two observations. The first one is 
\begin{lemma}\label{lemma.pair} Let ${\mathcal A} =  \langle A,\nu , X , d , d^\ast \rangle$ and ${\mathcal B} = \langle B , \eta , X , d , d^{\ast \ast}\rangle$ be two metric namings of the same space $\langle X , d\rangle$ and let $\phi:A \rightarrow B$ and $\psi:B \rightarrow A$ be continuous trackings of the identity function on $X$. Let $M$ be a realizer of the anti-Specker property of  $\langle X , d\rangle$ with respect to $\mathcal A$. Then, uniformly continuous in these data, we can construct a realizer $M'$ for the anti-Specker property of $\langle X , d \rangle$ with respect to $\mathcal B$.
\end{lemma}
\begin{proof}
Let $\{g_i\}_{i \in \N}$ be a sequence from $B^\ast$ avoiding all $g \in B$, and let $H:B \rightarrow \N$ be a realizer of this.

Let $M'(\{g_i\}_{i \in \N} , H) = M(\{\psi(g_i)\}_{i \in \N},\lambda f.H(\phi(f)))$. $M'$ will be as required. 
We leave the verification for the reader.
\end{proof}
Our naming of $X \times Y$ as a function of namings of $X$ and $Y$ is based on one particular pairing function on $\N^\N$. Another pairing function will give another construction of the naming of the product. Any natural choices of pairing functions will be equivalent to the extent that Lemma \ref{lemma.pair} applies for the respective namings of the product. This shows the independence of the actual choice of naming of the product. 

In order to prove the independence of the choice of metric, we must identify the kind of data we need for the transformation of one realizer $M$ of anti-Speckerhood to another.

Let $\langle A , \nu , X\rangle$ be a naming and let $d$ and $d_1$ be two metrics on $X$ such that $d_1$ is continuous in $d$. From elementary topology we know that if $\langle X,d\rangle$ is compact, then $\langle X , d_1 \rangle$ is compact, so anti-Speckerhood should be preserved from $\langle X , d \rangle$ to $\langle X , d_1\rangle$.
This can be done continuously in a realizer for the fact that $d_1$ is continuous in $d$, i.e. in a continuous functional $F:A \times \N \rightarrow \N$ satisfying
$$\forall f \in A \forall n \in \N  \forall g \in A(d(\nu(f),\nu(g)) <2^{- F(f,n)} \Rightarrow d_1(\nu(f),\nu(g)) < 2^{-n}).$$

\section{Refuting  that {\bf RPT} and  partially Cauchy implies Cauchy can be realized in  Kleene's second model}\label{Sec4}
\subsection{Two topological spaces}
The statement that a series of rational numbers is absolutely convergent is in essence a statement that another rational sequence, that is increasing and non-negative, is Cauchy. We are interested in Cauchy sequences of rational numbers in general, and in converging, non-negative, increasing sequences in particular in this section. Since we are dealing with a realizability model, we will consider the topologies on these sets induced by the naming via realizers of Cauchyness. In this subsection, we will analyze the properties of these topologies, using some of these properties in proving 1. and 2. from the introduction.
\begin{defi} Let $\Omega_2^*$ be the set of rational Cauchy sequences, and let $\Omega_1^*$ be the subset of non-negative increasing sequences (not necessarily strictly increasing).

We let $\Omega_2$ resp. $\Omega_1$ be the set of realizers for Cauchyness of these sequences. More precisely, $\Omega_i$ is the set of pairs $(f,\bar x)$ where $\bar x \in \Omega_i^*$ and $f$ is a modulus function for the sequence. These sets have topologies, called \emph{the Baire topologies}, inherited from $(\N^\N)^2$. We may also consider the Baire topology on the sets $\Omega_i^*$, a topology that we will see that is coarser than the one given below.

The maps $(f,\bar x) \mapsto \bar x$ restricted to the two sets $\Omega_i$ are the associated namings, and induce the quotient topologies on the sets $\Omega_i^*$. We will consider $\Omega_1^*$ and $\Omega_2^*$ as topological spaces with the quotient topologies, with $\Omega_1^* \subset \Omega_2^*$.\end{defi}
\begin{lemma} \label{lemma4.1}\hfill
\begin{enumerate}[label=\alph*)]
\item[a)] For $i \in \{ 1,2\}$ we have that any set $O \subseteq \Omega_i^*$ that is open in the Baire topology on $\Omega_i^*$ is also open in the quotient topology.
\item[b)]The topology on $\Omega_1^*$ is the one induced from the topology on $\Omega_2^*$.
\end{enumerate}
\end{lemma} 
\begin{proof}
a) is left for the reader, being a simple exercise in general topology. From a) it follows that $\Omega_1^*$ is a closed subset of $\Omega_2^*$ in the topology on $\Omega_2^*$, and then b) can also be left for the reader.
\end{proof}
\begin{lemma}
The map $\bar x \mapsto \lim \bar x$ is a continuous map from $\Omega_i^*$ to $\R$ for $i \in\{1,2\}$.
\end{lemma}
\proof The proof is easy, and is left for the reader.\qed

We will now prove a series of lemmas leading up to the observation that these spaces are metrizable. We will not need this fact in the sequel, but we will refer to some concepts defined in the process later.

 We will employ the following notation in this section: $\vec x$, $\vec y$ etc. will be finite sequences of rational numbers. $\sigma$, $\tau$ etc. will be finite, increasing, sequences from $\N$.
 \begin{defi} Given $\sigma$ and $\vec x$, let $K(\sigma,\vec x)$ be the set of extensions $\bar x \in \Omega_2^*$ of $\vec x$ with a modulus extending $\sigma$.\end{defi}
 \begin{lemma}\label{lemma.metric}\hfill
 \begin{enumerate}[label=\alph*)]
 \item[a)] Each set $K(\sigma,\vec x)$ is closed.
 \item[b)] If $K(\sigma,\vec x) \neq \emptyset$, then $K(\sigma , \vec x)$ has a non-empty interior $B(\sigma,\vec x)$.
 \item[c)] If $O \subseteq \Omega_2^*$ is open, and $\bar x \in O$, there is a pair $\sigma$ and $\vec x$ such that $$\bar x \in B(\sigma, \vec x) \subseteq K(\sigma , \vec x) \subseteq O.$$
 \end{enumerate}
 \end{lemma}
\begin{proof}\hfill
 \begin{enumerate}[label=\alph*)]
 \item[a)] Let $\bar x \not \in K(\sigma,\vec x)$. There can be two reasons for this. If $\bar x$ is not an extension of $\vec x$, then we use that the set of $\bar y$ that does not extend $\vec x$ is open in the Baire topology on $\Omega_2^*$, and then apply   Lemma \ref{lemma4.1}, a). If $\bar x$ has no modulus extending $\sigma$, this must be because there is an $n < lh(\sigma)$ and $i,j \geq  \sigma(n)$ such that $|x_i - x_j| \geq 2^{-n}$. Let $k > \max\{i,j, lh(\vec x)\}$. Then no extension of $$(x_0 , \ldots , x_k)$$ will be in $K(\sigma,\vec x)$, and the set of such extensions is open in the Baire  topology on $\Omega_2^*$. Again, we use Lemma \ref{lemma4.1},a) to conclude that $\bar x$ is in the interior of the complement of $K(\sigma,\vec x)$.
 \item[b)] Let $K(\sigma,\vec x)$ be nonempty, and let $\bar x$ be the extension of $\vec x$ where we just repeat the last element of $\vec x$. Clearly $\bar x \in K(\sigma,\vec x)$.

 For $n < lh(\sigma)$, the significance of $\sigma(n)$ is that $$\forall i,j \geq \sigma(n)(|x_i - x_j| < 2^{-n}).$$
 We now let $O$ be the set of extensions $\bar y$ of $\vec x$ in $\Omega_2^*$ such that
 $$\exists \epsilon > 0 \forall n < lh(\sigma)\forall i,j \geq \sigma(n)(|y_i - y_j| < 2^{-n} - \epsilon).$$
 If $(f,\bar y) \in \Omega_2$ we let $(f,\bar y) \in U$ if $\bar y$ extends $\vec x$ and for some $\epsilon > 0$ and some $k$:
 \begin{itemize}
 \item[-] $2^{-k} < \epsilon$.
 \item[-] $\forall n < lh(\sigma) \forall i,j(n \leq i,j \leq f(k) \Rightarrow |y_i - y_j| < 2^{-n} - 2\epsilon)$.
 \end{itemize}
 We have that $(f,\bar y) \in U \Leftrightarrow \bar y \in O$ and that $\bar x \in O$.
 Clearly $U$ is open in the Baire topology, so $O\subset K(\sigma , \vec x)$ is open in $\Omega_2^*$. It follows that $\bar x$ is in the interior of $K(\sigma,\vec x)$.
\item[c)] is trivial, and is left for the reader.\qedhere
\end{enumerate}
\end{proof}
 \begin{rem} Our proof of a)  shows that $K(\sigma,\vec x)$ is closed in the Baire  topology. In the proof of b) we  construct the  set $O$. This set is actually the interior $B(\sigma,\vec x)$ of $K(\sigma,\vec x)$. We will not need this and leave the proof  as an exercise. 
 
 The set of sets $K(\sigma,\vec x)$ is a pseudo-base as defined by Schr\"oder \cite{Schroder}, and c) reflects one of the basic properties of pseudo-bases.\end{rem}
  \begin{cor}\label{corollary.metric}
 The spaces $\Omega_1^*$ and $\Omega_2^*$ are metrizable.
 \end{cor}
 \begin{proof}
 It suffices to prove this for $\Omega_2^*$.
 By Lemma \ref{lemma.metric} and the fact that there are only countably many pairs $(\sigma,\vec x)$ we use standard elementary topology to see that $\Omega_2^*$ is a regular space, and then the corollary is a consequence of the Urysohn Metrization Theorem.\end{proof}
%
In a topological space, we let a set $C$ be \emph{clopen} if it is both closed and open. A space is \emph{zero-dimensional} if it has a basis of clopen sets. The Baire space, and thus every subspace of the Baire space, is zero-dimensional. Since the Baire topology on $\Omega_i^*$ is a sub topology of the real topology, this real topology has a lot of clopen sets. We will show that nevertheless, $\Omega_1^*$ is not zero-dimensional, and thus, $\Omega_2^*$ is not zero-dimensional. We will use this to prove the main results of this section.
\begin{lemma}\label{lemma4}
Let $\{\bar x_n\}_{n \in \N}$ be a sequence from $\Omega_1^\ast$ such that
\begin{enumerate}[label=\roman*)]
\item[i)] The sequence $\{x_{n,i}\}_{n \in \N}$ is eventually constant for each $i \in \N$.
\item[ii)] The sequence $\{\lim \bar x_n\}_{n \in \N}$ is increasing and bounded.
\end{enumerate}
Then the sequence $\{\bar x_n\}_{n \in \N}$ converges in $\Omega_2^\ast$ to its pointwise limit.
\end{lemma}
\begin{proof}
Let $\bar x$ be the pointwise limit. Since each $\bar x_n$ is increasing, we also have that $\bar x$ is increasing, and we have that $$\lim \bar x = \lim_{n \rinf} \lim \bar x_n.$$
We will show that there is one common modulus  for $\bar x$ and each $\bar x_n$.

Let $k$ be given, and let $i_0$ be such that if $i,j \geq i_0$ then $|x_i - x_j| < 2^{-(k+1)}$.
Let $n_0$ be so large that for all $i \leq i_0$ and all $n \geq n_0$ we have that $x_{n,i} = x_{n_0,i}$ and let $i_1 \geq i_0$ be so large that for all $n \leq n_0$ and all $i,j \geq i_1$ we have that $|x_{n,i} - x_{n,j}| < 2^{-k}$.

If $n > n_0$ and $i_0 \leq i < j$ we use that we only consider increasing sequences, and have
$$x_{n,j} - x_{n,i} \leq x_{n,j} - x_{n,i_0} \leq \lim \bar x_n - x_{n,i_0} \leq \lim \bar x - x_{n,i_0}= \lim \bar x - x_{i_0} < 2^{-k}.$$
So, for all $n$ and all $i,j \geq i_1$ with $i < j$ we have that $x_{n,j} - x_{n,i} < 2^{-k}$.
Since $k$ was arbitrary, we use this to define one common modulus  for $\bar x$ and all $\bar x_n$. This proves that $\bar x = \lim_{n \rinf} \bar x_n$ in the sense of $\Omega_1^\ast$, and thus in the sense of $\Omega_2^\ast$.\end{proof}
\begin{lemma}\label{lemma6}
Let $A \subseteq \Omega_1^\ast$ be open, and let $\bar x \in A$.
Then there is an $\epsilon > 0$ such that for each $n$ there is an $\bar x_n \in A$ satisfying
\begin{enumerate}[label=\roman*)]
\item[i)] $x_i = x_{n,i}$  whenever $i \leq n$.
\item[ii)] $\lim \bar x_n > \lim \bar x + \epsilon$.
\end{enumerate}
\end{lemma}
\begin{proof}
Assume this is not the case.  Then there will be a modulus $f:\N \rightarrow \N$    for $\bar x$ such that for all $n$ and all $\bar y \in A$, if the sequences $\bar x$ and $\bar y$ are equal for the first $f(n)$ items, then $\lim \bar y < \lim \bar x + 2^{-n}$. 

Let $g$ be defined by $g(n) = f(n+1)$. $g$ will also be a modulus for $\bar x$. Since $A$ is open, there is a number $k$ such that $\bar y \in A$ whenever $\bar y$ has a modulus that agrees with $g$ on the first $k$ inputs and $\bar y$ agrees with $\bar x$ for the first $g(k-1)$ items. With our notation,
$$K(\langle x_0 , \ldots , x_{g(k-1) - 1}\rangle, \bar g(k)) \subseteq A\;,$$ which is the same as
$$K = K(\langle x_0 , \ldots , x_{f(k)) - 1}\rangle, \bar g(k)) \subseteq A.$$
We define the sequence $\bar y$ by
\begin{itemize}
\item[-] $y_i = x_i$ if $i < f(k)$
\item[-] $y_i = x_i + 2^{-k}$ if $i \geq f(k)$.
\end{itemize}
$g$ will be a modulus for $\bar y$ as well. This is verified by a simple calculation. Then $\bar y \in K \subseteq A$.

But $\bar x$ and $\bar y$ agrees on the first $f(k)$ items, while $\lim \bar y = \lim \bar x + 2^{-k}$, and this contradicts our choice  of $f$.
This contradiction shows that we have disproved our assumption, and the lemma is proved.\end{proof}
\begin{lemma}\label{lemma5}
Let $A \subseteq \Omega_1^\ast$ be nonempty and clopen.
Then $\{\lim \bar x \mid \bar x \in A\}$ is unbounded.
\end{lemma}
\begin{proof}Assume that $\{\lim \bar x \mid \bar x \in A\}$ is bounded, and let $a_0$ be the supremum of this set.
Let $A_0 = A$.
Let $\bar x_0 \in A$ be such that $a_0 - \lim \bar x_0 < 1$ and let $$A_1 = \{\bar x \in A_0 \mid x_0 = x_{0,0}\}.$$
Let $a_1$ be the supremum of
$$\{\lim \bar x \mid \bar x \in A_1\}.$$
By Lemma \ref{lemma6} we have that $a_1 > \lim \bar x_0$
Let $\bar x_1 \in A_1$ be such that $\lim \bar x_0  < \lim \bar x_1$ and $a_1 - \lim \bar x_1 < 2^{-1}$. Then let $$A_2 = \{ \bar x \in A_1 \mid x_1 = x_{1,1}\}.$$
We continue this construction by recursion.

Assume that $\bar x_n$ and $A_n$ are constructed, where $A_n$ is clopen.
Let $$A_{n+1} = \{\bar x \in A_n \mid x_n = x_{n,n}\}$$ and let $a_{n+1}$ be the supremum of $\{\lim \bar x \mid \bar x \in A_{n+1}\}$.
By Lemma \ref{lemma6} we have that $a_{n+1} > \lim \bar x_n$.

Let $\bar x_{n+1} \in A_{n+1}$ be such that $\lim \bar x_n < \lim \bar x_{n+1}$ and such that $a_{n+1} - \lim \bar x_{n+1} < 2^{-(n+1)}$.
This ends the recursion step.

The sequence $\{\bar x_n\}$ satisfies by construction the requirements of Lemma \ref{lemma4}. Let $\bar x = \lim_{n \rinf }\bar x_n \in A$ (here we use that $A$ is closed).
But by Lemma \ref{lemma6}, there will be an $\epsilon > 0$ such that each $(x_0 , \ldots , x_n)$ has an extension $\bar x'$ in $A$ with $\lim \bar x' - \lim \bar x > \epsilon$. 
Here we use that $A$ is open. This actually means that for each $n$ there is an $\bar x' \in A_{n+1}$ such that $\lim \bar x' - \lim \bar s \geq \epsilon$. This is, however, in conflict with the construction of the sequence $\bar x$, since we ensure that $\lim \bar x = \lim_{n \rinf} a_n$,  and represents a contradiction. The assumption was that the lemma is false, so the lemma is proved.
\end{proof}
\begin{lemma}\label{lemma19}
$\Omega_1^\ast$ is not zero-dimensional.
\end{lemma}
\begin{proof}
Since there are open sets, e.~g.~ $\{\bar x \in \Omega_1^\ast \mid 0 < \lim \bar x < 1\}$, that do not satisfy the conclusion of Lemma \ref{lemma5}, there are nonempty open sets that  have no nonempty clopen subsets.\end{proof}
\begin{cor}\label{corollary.vital} For $i \in \{1,2\}$ there is no continuous function $\phi:\Omega_i^* \rightarrow \N^\N$ such that $(\phi(\bar x) , \bar x) \in \Omega_i$ for each $\bar x \in \Omega_i^*$.
\end{cor}
\begin{proof}
Assume that there is a continuous $\phi$ with the mentioned property. We will obtain a contradiction by showing that $\Omega_i^*$ is zero-dimensional.

Let $O \subseteq \Omega_i^*$ be open and let $\bar x \in O$. Then there is an open set $U \subseteq \Omega_i$ such that $$(f, \bar y) \in U \Leftrightarrow \bar y \in O$$ for all $(f,\bar y) \in \Omega_i$.
The Baire topology is zero-dimensional, so let $B \subseteq U$ be clopen such that $$(\phi(\bar x),\bar x) \in B.$$
Then $C \subseteq O$ defined by
$$\bar y \in C \Leftrightarrow (\phi(\bar y) , \bar y) \in B$$ is a clopen subset of $O$ containing $\bar x$.\end{proof}
\subsection{The proofs}\label{RPT-proof}
\begin{thm}
The statement {\bf pC} $\Rightarrow$ {\bf C} is not realizable in the extensional model of continuous functionals.
\end{thm}
\begin{proof}
The proof will be contrapositive, we assume that the statement is realizable and obtain a contradiction to  Corollary \ref{corollary.vital}.

Assume that $M$ is continuous such that for every $\bar x \in
\Omega_2^*$ and every $F$ that realizes that $\bar x$ is partially
Cauchy we have that $M(F,\bar x)$ is a modulus for $\bar x$. Given
$\bar x$, $g \geq id$ and $n \in \N$, we let $$F_{\bar x}(g )(n) = \mu
k.\forall m \geq k (diam\{x_m , \ldots , x_{g(m)}\} <
2^{-n}).$$\medskip

\noindent{\em Claim.} The map $\bar x \mapsto F_{\bar x}$ is continuous.
\medskip

\noindent\emph{Proof of Claim.}
It suffices to prove that the map $(f,\bar x) \mapsto F_{\bar x}$ is continuous on $\Omega_2$.

Given $\bar x$, a modulus $f$ for $\bar x$, $g \geq id$  and $n$ it is easy to find, in a continuous way, an upper bound for $F_{\bar x}(g)(n)$, and from this upper bound, we can find the exact value. The claim follows.

Using the claim, we let $\phi(\bar x) = M(F_{\bar x}, \bar x)$, contradicting Corollary \ref{corollary.vital} for $i = 2$.
This proves the theorem  by contradiction.
\end{proof}
The refutation of {\bf RPT} in our model follows the same pattern, but is a bit more elaborate. We need a few lemmas before we can contradict Corollary \ref{corollary.vital} from the assumption that {\bf RPT} is realizable. 

The rest of this section is a proof of
\begin{thm}\label{theoremRPT} The Riemann Permutation Theorem cannot be realized in the extensional realizability model induced by $K_2$.
\end{thm}
The key technical lemma will be:

\begin{lemma}\label{lemma8}
Let $\bar x$ and $\bar b$ be  sequences of non-negative rational numbers such that $x_i > 0$ for infinitely many $i \in \N$.
Uniformly computable in $\bar x$ and $\bar b$ we can find
\begin{itemize}
\item[-] $k_i \geq 1$ for each $i \in \N$
\item[-] $y_{i,j} \in \Q$ for each $i \in \N$ and $j$ with $1 \leq j \leq k_i$, 
\end{itemize}
defining $B$ to be the set $$B = \{(i,j)\mid i \in \N \wedge 1 \leq j \leq k_i\},$$
such that 
\begin{itemize}
\item[-] $x_i = \sum_{j = 1}^{k_i}|y_{i,j}|$ for each $i \in \N$.
\item[-]If $A \subseteq B$ is finite and $n \in \N$ then
$$\left |\sum_{(i,j) \in B \setminus A}y_{i,j}\right | \neq b_n,$$
\end{itemize}
where we use the lexicographical ordering of $B$ in viewing $\sum_{(i,j) \in B \setminus A}y_{i,j}$ as a series.
The inequalities are considered to be fulfilled when $\sum_{(i,j) \in B}y_{i,j}$ is diverging.
\end{lemma}
\begin{proof}
We will construct $k_s$ and $y_{s,1} , \ldots , y_{s,k_s}$ by recursion in \emph{stages} $s$. 
During the construction, we  let $$B_s = \{(i,j) \mid i < s \wedge 1 \leq j \leq k_i\},$$ ordered lexicographically. 
We let $R_s$ be the set of pairs $(A,n)$ where $A \subseteq B_s$ and $n \leq s$. This defines an increasing family of finite sets, and $R = \bigcup_{s \in \N}R_s$ will at the end of the construction consist of all pairs $(A,n)$ where $A \subset B$ is finite and $n \in \N$.

During the construction, we will set up a \emph{protection} $r_{(A,n)}$ for each $(A,n) \in R$. A protection will be a positive rational number, and when it is set up , it will be kept through the rest of the construction.  During the construction we will ensure that when there is a protection $r_{(A,n)}$ for $(A,n)$ at the end of stage $s$, then 
$$(\dagger)\hspace{5mm}\left | \; \left | \sum_{(i,j) \in B_{s+1} \setminus A}y_{i,j}\;\right |  - b_n\;\right | > r_{(A,n)}\;,$$
which will in turn secure that at the end of the construction 
$$\left | \;\left | \sum_{(i,j) \in B \setminus A}y_{i,j}\;\right | - b_n \;\right | \geq r_{(A,n)} > 0.$$
We will also ensure that at the end of a stage $s$ where $x_s > 0$, all $(A,n) \in R_s$ have protections.

Let us now go to the details of the construction.
If $x_s = 0$, we let $k_s = 1$ and $y_{s,1} = 0$, and we move on to the next stage.

So assume that $x_s > 0$ and let $(A,n) \in R_s$. There are three possibilities:
\begin{enumerate}
\item There is a protection $r_{(A,n)}$ for $(A,n)$ at the beginning of stage $s$ .
\item There is no protection for $(A,n)$ at the beginning of stage $s$  and $$\left |\;\sum_{(i,j)\in B_s \setminus A}y_{i,j}\; \right| \neq b_n.$$
\item There is no protection for $(A,n)$ at the beginning of stage $s$,  and $$\left |\;\sum_{(i,j)\in B_s \setminus A}y_{i,j}\; \right | = b_n.$$
\end{enumerate}
Observe that (1) will be the case exactly when $(A,n) \in R_t$ for some $t < s$ where $x_t > 0$.

For $(A,n)$ in Case (2), we let  $$r_{(A,n)} = \frac{1}{2} \left |\;\left |\sum_{(i,j)\in B_s \setminus A}y_{i,j}\;\right | - b_n\;\right | $$ be the protection. Then $(\dagger)$ will hold for this $(A,n)$ at this stage of the stage.

Now, we will find $y_{s,1} , \ldots , y_{s,k_s}$ such that 
\begin{itemize}
\item[-] $x_s = \sum_{j = 1}^{k_s}|y_{s,j}|.$
\item[-] For $(A,n)$ in cases (1) and (2), $(\dagger)$ will be preserved if we extend our series with $y_{s,1}, \ldots , y_{s,k_s}$. 

Notice that $B_{s+1}$ will be $B_s$ end-extended with $(s,1) , \ldots , (s,k_s)$ and that $A \subseteq B_s$, so we may do so by ensuring that the absolute value of $$\sum_{j = 1}^{k_s}y_{s,j}$$ is smaller than the difference between $r_n$ and $$\left | \; \left | \sum_{(i,j) \in B_{s} \setminus A}y_{i,j}\;\right |  - b_n\;\right | $$ and then apply the triangle inequality.
\end{itemize}

\noindent Let $t$ be the minimal value of $$\left | \;\left |  \sum_{(i,j) \in B_s \setminus A}y_{i,j}  \;\right | -  b_n\;\right |- r_{(A,n)}$$ for all $(A,n) \in R_s$ having a protection after we added those from case (2).
Let $k_s$ be an uneven number such that $$\frac{x_s}{k_s} < \frac{t}{2}.$$
We let $y_{s,j} = (-1)^{j+1}\frac{x_s}{k_s}$ for $1 \leq j \leq k_s$. Then $$\sum_{j = 1}^{k_s}y_{s,j} = \frac{x_s}{k_s}\;,$$
and the property $(\dagger)$ will be preserved.

If we are in Case (3) for  $(A,n) \in R_s$, we let $r_{(A,n)} = \frac{x_s}{2k_s}$.
Since in this case$$ \left | \; \left | \sum_{(i,j) \in B_{s+1} \setminus A}y_{i,j}\;\right |  - b_n\;\right | = \frac{x_s}{k_s}$$  this will also satisfy $(\dagger)$ after this stage.  
Then the construction proceeds to the next stage. 

The properties required by the lemma are secured in this construction, so the proof is complete.
\end{proof}
We now let $\bar a$ be a convergent, increasing sequence of non-negative rational numbers that is not almost constant, and let 
\begin{itemize}
\item[-] $x_0 = a_0$
\item[-] $x_{i+1} = a_{i+1} - a_i$
\end{itemize}
Let $\bar z$ be the sequence $$y_{0,1} , \ldots , y_{0,k_0},y_{1,1}, \ldots , y_{1,k_1}, \ldots$$
as constructed from $\bar x$ in the proof of Lemma \ref{lemma8}, where we use $b_n = 2^{-n}$.

We know that $$\sum_{i \in \N}|z_i| = \lim_{j \rinf }a_j,$$ so the series is absolutely convergent.
\begin{defi} Let $p:\N \rightarrow \N$ be a permutation of $\N$.
Let $F_{\bar a}(p,n)$ be the least $m$ such that $$\left |\sum_{k = i}^j z_{p(k)}\right | < 2^{-n}$$ whenever $m \leq i \leq j$.
\end{defi}
\begin{lemma}\label{lemma9}
$F_{\bar a}(p,n)$ is uniformly computable from $\bar a$, $p$, $n$ and an arbitrary modulus function $f$ for $\bar a$.
\end{lemma}
\begin{proof}
It is of course important that $\bar z$ has the property of the conclusion in Lemma \ref{lemma8} with respect to $\{2^{-n}\}_{n \in \N}$.
Since $$\sum_{k = 0}^\infty z_k$$ is absolute convergent, we know that the sum $$\sum_{k = 0}^\infty z_{p(k)}$$ is independent of $p$, so for any $m$, $n$ and permutation $p$ we have that
$$\left |\sum_{k > m} z_{p(k)}\;\right | \neq 2^{-n}.$$
Given $m$, $n$, $p$ and $f$ we will show that we can effectively distinguish between the two cases
\begin{enumerate}[label=\uppercase{\roman*}]
\item[I] $\exists i \geq m \exists j \geq i (|\sum_{k = i}^j z_{p(k)}| \geq 2^{-n})$
\item[II] otherwise.
\end{enumerate}
In Case I, this is simply witnessed by $i$, $j$ and a finite fragment of $p$, we do not rely on $f$ or Lemma \ref{lemma8} in this case.

If we are in Case II, we know, by Lemma \ref{lemma8}, that $$\left |\sum_{k \geq m}z_{p(k)}\;\right | < 2^{-n}.$$
(The sum cannot be larger than $2^{-n}$, since we then would be in Case I.)
Then there is a rational number $r > 0$ such that for all $i \geq m$ and all $j \geq i$ we have that
$$\left |\sum_{k = i}^j z_k\;\right | < 2^{-n} - r.$$
Let $n_0$ be such that $2^{-n_0} < r$  and use $f$ to find $n_1$ such that $$\sum_{k \geq n_1}\left |z_k\;\right |< 2^{-n_0}.$$
Let $p_0$ be an initial segment of $p$ such that all $y_{i,j}$ where $i < n_1$ are in the range of $p_0$, and let $k_0$ be the length of the initial segment $p_0$.

We need the number $r$, which is not available from our data unless we already know that we are in Case II, in order to carry out this explicit construction, but the consequence we will make use of is that if we are in Case II, there are numbers $n_0$, $n_1 = f(n_0)$ and $k_0$  such that we effectively from $f$, $p$ restricted to $k_0$  and the rest of the data know that
\begin{enumerate}
\item If $m \leq i \leq j < k_0$ then $$\left |\sum_{k = i}^j z_{p(k)}\;\right | + 2^{-n_0} < 2^{-n}.$$
\item $$\sum_{k \geq k_0}| z_{p(k)}| < 2^{-n_0}.$$
\end{enumerate}
From the existence of $n_0$, $n_1$ and $k_0$ with the properties given above, we can deduce that we are in Case II.

We can then, uniformly in $n$, $f$, $p$ and $m$ search for a witness, either for Case I or for Case II, and we are bound to find one. Thus we can, uniformly in $p$, $f$ and $m$ decide between the two cases.

We use the fact that we can split between the two cases in order to compute $F_{\bar a}(p,n)$, where $F_{\bar a}(p,n)$ is the least $m$ such that we are in Case II.
This proves the Lemma.\end{proof}
\begin{lemma}
Realizability of {\bf RPT} contradicts Corollary \ref{corollary.vital} with $i = 1$.
\end{lemma}
\begin{proof}
Assume That {\bf RPT} can be realized. 
Let $\bar a$ be an increasing, convergent sequence of non-negative rational numbers, and let $\bar x$, $\bar z$ and $F_{\bar a}$ be as in the proofs of Lemmas \ref{lemma19} and \ref{lemma9}.

$F_{\bar a}$ is uniformly continuous in $\bar a$ and any modulus $f$ for $\bar a$, but independent of the choice of $f$. By the assumption, we can continuously find a modulus $g$ for $$\sum_{k = 0}^{\infty }|z_k|$$ from $F_{\bar a}$.
From $g$ we can continuously find a modulus $f_g$ for $\bar a$, and $f_g$ will be independent of $f$. This shows that Corollary \ref{corollary.vital} fails.\end{proof}
\vspace{2mm}
We have indeed proved Corollary \ref{corollary.vital}, so this also completes the proof of Theorem \ref{theoremRPT}

\section{Conclusion}\label{Sec5}
We have decided the truth value of three weakly non-constructive statements in one fixed realizability model. $K_2$ is of course an important example, but it is still just one example, so the amount of information we can draw from our results will be limited.

One obvious consequence is that {\bf RPT} is not deducible,  with constructive means, from any of the facts realized by the model, including the fan theorem, bar induction and dependent choice, or in any formal theory with an extensional realizability model based on $K_2$.

Our analysis of the anti-Specker property is really by introducing the compactness base. A compactness base is, loosely spoken, a realizer of the fact that the space is compact, and we reduced the question of closure under products for anti-Specker spaces to closure under products of compact spaces for this model. Since the model realizes the fan theorem, it should not come as a surprise that the equivalence between being compact and being anti-Specker can be established.

The natural question now is if it is possible to prove this closure under products from  principles known to be realized by extensional $K_2$, such as the fan theorem and bar recursion. We leave this problem for future investigations.


\begin{thebibliography}{99}
%
\bibitem{BSS} I.~ Battenfeld, M.~ Schr\"oder and A.~ Simpson, {\em A Convenient Category of Domains}, in L.~ Cardelli, M. ~Fiore and G.~ Winskel (eds.) Computation, Meaning and Logic, Articles dedicated to Gordon Plotkin, Electronic Notes in Computer Science 34 (2007).
%
\bibitem{BD1} J.~Berger and D.~Bridges, \emph{Rearranging series constructively}, Journal of Universal Computer Science 15 (2009) pp. 3160 - 3168.
%
\bibitem{BD2} J.~Berger, D.~Bridges, H.~Diener and H.~Schwichtenberg, \emph{Constructive aspects of Riemann's permutation theorem for series}, unpublished.
%
\bibitem{Bridges} D.~ Bridges, \emph{Inheriting the anti-Specker property} preprint, University of Canterbury, New Zealand, 2008, submitted.
%
\bibitem{Ishihara1}H. Ishihara, \emph{Continuity and nondiscontinuity in constructive mathematics}, Journal of Symbolic Logic 56 (1991) pp. 1349-1354.
%
\bibitem{Ishihara2}H.~Ishihara, \emph{Continuity properties in constructive mathematics}, Journal of Symbolic Logic 57 (1992) pp. 557-565.
%
\bibitem{Kleene.59II} S.C. Kleene, Countable functionals, in A. Heyting (ed.) \emph{Constructivity in Mathematics}, North-Holland, Amsterdam, pp. 81 - 100 (1959).
%
\bibitem{LS}P.~ Lietz and T.~ Streicher
\emph{Realizability models refuting Ishihara's boundedness principle},  
Ann. Pure Appl. Logic  163  (2012),  no. 12, pp. 1803-1807.
%
\bibitem{LD} R.~ S.~ Lubarsky and H.~ Diener, \emph{Principles weaker than BD-N}. Journal of Symbolic Logic 78 (2013), pp. 873 - 885.
%
\bibitem{Schroder} M.~ Schr\"oder , \emph{Extended admissibility}, Theoretical Computer Science 284 (2002), pp. 519-538.
%
\bibitem{Weihrauch} K.~Weihrauch, {\em  Computational Analysis}, Texts in Theoretical Computer Science, Springer Verlag (2000).

\end{thebibliography}
\end{document}